\documentclass[11pt]{article}
\usepackage[margin=1in]{geometry}
\usepackage{amsfonts,latexsym,amsthm,amssymb,amsmath,amscd,euscript,tikz,mathtools}
\usepackage{framed}
\usepackage{hyperref}
\hypersetup{colorlinks=true,citecolor=blue,urlcolor =black,linkbordercolor={1 0 0}}
\usepackage{enumitem}
\usepackage{siunitx}
\usepackage{textcomp}
\usepackage{physics}
\usepackage{mathrsfs}
\usepackage{dsfont}
\usepackage[ruled,vlined]{algorithm2e}
\usepackage{titlesec}
\usepackage{tikz-cd}

\allowdisplaybreaks[1]

\newtheorem{theorem}{Theorem}
\newtheorem{proposition}[theorem]{Proposition}
\newtheorem{lemma}[theorem]{Lemma}
\newtheorem{corollary}[theorem]{Corollary}

\theoremstyle{definition}
\newtheorem{definition}[theorem]{Definition}

\theoremstyle{remark}

\newcommand{\cK}{\mathcal{K}}

\newcommand{\cX}{\mathcal{X}}

\newcommand{\bN}{\mathbb{N}}

\newcommand{\bR}{\mathbb{R}}

\newcommand{\1}{\mathds{1}}

\newcommand{\Wu}{W^{\uparrow}}
\newcommand{\Wd}{W^{\downarrow}}
\newcommand{\ud}{\uparrow\downarrow}
\newcommand{\du}{\downarrow\uparrow}
\newcommand{\Wud}{W^{\ud}}
\newcommand{\Wdu}{W^{\du}}

\newcommand{\nc}{\newcommand}

\nc{\on}{\operatorname}
\nc{\Spec}{\on{Spec}}
\nc{\Aut}{\textit{Aut}}
\nc{\id}{\textit{id}}
\nc{\chr}{\on{char}}
\nc{\im}{\on{im}}
\nc{\Hom}{\on{Hom}}
\nc{\lcm}{\on{lcm}}
\nc{\dual}[1]{\prescript{t}{}{#1}}
\nc{\transpose}[1]{{#1}^{\intercal}}
\nc{\Sym}{\on{Sym}}
\nc{\End}{\on{End}}
\nc{\stab}{\on{stab}}
\nc{\Li}{\on{Li}}
\nc{\spn}{\on{span}}
\nc{\sgn}{\on{sgn}}
\nc{\supp}{\on{supp}}
\nc{\Unif}{\on{Unif}}

\title{Improved Product-Based High-Dimensional Expanders}
\author{Louis Golowich\thanks{Harvard University. Email: \href{mailto:lgolowich@college.harvard.edu}{lgolowich@college.harvard.edu}}}
\date{July 3, 2021}

\begin{document}
\maketitle


\begin{abstract}
  High-dimensional expanders generalize the notion of expander graphs to higher-dimensional simplicial complexes. In contrast to expander graphs, only a handful of high-dimensional expander constructions have been proposed, and no elementary combinatorial construction with near-optimal expansion is known. In this paper, we introduce an improved combinatorial high-dimensional expander construction, by modifying a previous construction of Liu, Mohanty, and Yang (ITCS 2020), which is based on a high-dimensional variant of a tensor product. Our construction achieves a spectral gap of $\Omega(\frac{1}{k^2})$ for random walks on the $k$-dimensional faces, which is only quadratically worse than the optimal bound of $\Theta(\frac{1}{k})$. Previous combinatorial constructions, including that of Liu, Mohanty, and Yang, only achieved a spectral gap that is exponentially small in $k$. We also present reasoning that suggests our construction is optimal among similar product-based constructions.
\end{abstract}

\section{Introduction}
Graphs that are sparse but well connected, called expander graphs, have numerous applications in various areas of computer science (see for example \cite{hoory_expander_2006}). Recently, there has been much interest in generalizing the notion of expansion to higher-dimensional simplicial complexes, beginning with the work of Lineal and Meshulam~\cite{linial_homological_2006}, Meshulam and Wallach~\cite{meshulam_homological_2009}, and Gromov~\cite{gromov_singularities_2010}. While multiple notions of high-dimensional expansion have been introduced (see the survey by Lubotzky~\cite{lubotzky_high_2018}), these notions agree in that random simplicial complexes are not good high-dimensional expanders. In contrast, ordinary random graphs are near-optimal expanders with high probability. Therefore constructions of high-dimensional expanders are of particular interest. Early constructions of high-dimensional expanders, namely Ramanujan complexes \cite{lubotzky_ramanujan_2005,lubotzky_explicit_2005}, were quite mathematically involved, whereas more simple and combinatorial constructions have been introduced recently.

In this paper, we introduce a modification of the high-dimensional expander construction of Liu, Mohanty, and Yang~\cite{liu_high-dimensional_2020}, which is based on a sort of high-dimensional tensor product. We then show that this modification gives an exponential improvement in the dependence of the $k$th order spectral gap on the dimension $k$, and we discuss why our results suggest this modification is optimal among constructions with the same general product structure. These results also address a question posed by Liu et al.~\cite{liu_high-dimensional_2020} pertaining to the limitations of their product-based construction.

\subsection{High-dimensional expanders}
A high-dimensional expander is a simplicial complex with certain expansion properties. A \textbf{simplicial complex} is a hypergraph with downward-closed hyperedges, called faces. That is, a simplicial complex $X$ on $n$ vertices is a collection of faces $X\subset 2^{[n]}$, where for any face $\sigma\in X$, all subsets of $\sigma$ are also faces in $X$. The \textbf{dimension} of a face $\sigma$ is $\dim(\sigma)=|\sigma|-1$, and the dimension of $X$ is the maximum dimension of any face in $X$. The \textbf{1-skeleton} of a simplicial complex is the undirected $n$-vertex graph whose edges are given by the 1-dimensional faces. We restict attention to \textbf{pure} simplicial complexes, meaning that every face is contained in a face of maximal dimension.

We consider the notion of high-dimensional expansion introduced by Kaufman and Mass \cite{kaufman_high_2017}, which requires rapid mixing for the high-order ``up-down'' or ``down-up'' random walks on simplicial complexes. The $k$-dimensional up-down walk on a simplicial complex specifies transition probabilities for a walk that alternates between faces of dimension $k$ and $k+1$. The 0-dimensional up-down walk is an ordinary lazy random walk on the 1-skeleton of the simplicial complex. Therefore 1-dimensional expanders are just ordinary (spectral) expander graphs. The spectral gaps of higher-order walks are difficult to bound directly, so they are instead typically bounded using a line of work \cite{kaufman_high_2017,dinur_high_2017,kaufman_high_2018,alev_improved_2020}, which has shown that a large spectral gap for high-order walks is implied by good local expansion, that is, good spectral expansion of a specific set of graphs that describe the local structure of the simplicial complex. Formal definitions of high-order walks and local expansion are provided in Section~\ref{sec:prelim}.

We are interested in constructions of infinite families of high-dimensional expanders with bounded degree and spectral gap for all dimensions. Specifically, for any fixed $H\geq 1$, a \textbf{$H$-dimensional expander family} is a family $\cX$ of $H$-dimensional simplicial complexes such that there is no finite upper bound on the number of vertices of elements $X\in\cX$, and the following two properties hold:
\begin{enumerate}
\item Bounded degree: There exists some $\overline{d}>0$ such that for every $X\in\cX$, each vertex in $X$ belongs to at most $\overline{d}$ faces.
\item Bounded spectral gap: There exists some $\underline{\nu}>0$ such that for every $X\in\cX$ and every $0\leq k\leq H-1$, the $k$-dimensional up-down walk on $X$ has spectral gap $\geq\underline{\nu}$.
\end{enumerate}
In general, the spectral gap of the $k$-dimensional up down walk cannot be greater than $\frac{2}{k+2}$ (see for example Proposition~3.3 of \cite{alev_improved_2020}). Our goal is to prove good lower bounds for this spectral gap for specific constructions of $H$-dimensional expander families. That is, we are interested in the optimal relationship between dimension and spectral gap, and only require an arbitrary upper bound on degree. This goal differs from the study of expander graphs, and specifically Ramanujan graphs, which focuses on the optimal relationship between degree and spectral gap.

While Kaufman and Mass~\cite{kaufman_high_2017} showed that Ramanujan complexes are high-dimensional expanders, multiple more elementary constructions have since been introduced \cite{conlon_hypergraph_2019,conlon_hypergraph_2020,chapman_expander_2020,liu_high-dimensional_2020,kaufman_high_2020,friedgut_hyper-regular_2020}. However, only three of these constructions \cite{liu_high-dimensional_2020,kaufman_high_2020,friedgut_hyper-regular_2020} provide constant-degree high-dimensional expander families of all dimensions. The construction of Kaufman and Oppenheim~\cite{kaufman_high_2020} and the hyper-regular variant introduced by Friedgut and Iluz~\cite{friedgut_hyper-regular_2020} are based on coset geometries, and achieve near-optimal expansion in all dimensions. In contrast, the construction of Liu et al.~\cite{liu_high-dimensional_2020} is much more elementary, as it consists of a sort of high-dimensional tensor product between an expander graph and a constant-sized complete simplicial complex. However, this construction has suboptimal expansion in high dimensions. Specifically, Alev and Lau~\cite{alev_improved_2020} showed that the $k$-dimensional up-down walk on the $H$-dimensional construction of Liu et al.~\cite{liu_high-dimensional_2020} has spectral gap at least $\frac{c}{(k+2)2^{k+2}}$, where $1\leq c\leq 2$ is a constant depending on $H$ and on the expander graph used in the construction. Note that this bound has exponential dependence on $k$, in contrast to the optimal linear dependence $\Omega(\frac{1}{k})$.

\subsection{Contributions}
In this paper, we present a modification of the high-dimensional expander family of Liu et al.~\cite{liu_high-dimensional_2020}, for which we show that the spectral gap of the $k$-dimensional up-down walk is at least $\frac{1}{(1+\log H)(k+2)(k+1)}$. This quadratic dependence on $k$ provides an exponential improvement compared to the spectral gap bound of $\frac{c}{(k+2)2^{k+2}}$ for the construction of Liu et al.~\cite{liu_high-dimensional_2020}. We attain this exponential improvement using the same product structure as Liu et al.~\cite{liu_high-dimensional_2020} while adjusting the weights of faces. Our modified construction also yields improved local expansion in high dimensions. For every $1\leq k\leq H-2$, we show that the 1-skeleton of the link of any $k$-dimensional face in our construction has spectral gap at least $\frac{k+1}{k+2}$, an improvement over the analagous bound of $\frac12$ for the construction of Liu et al.~\cite{liu_high-dimensional_2020}.

The organization of the remainder of this paper is as follows. Section~\ref{sec:prelim} presents preliminary notions, and Section~\ref{sec:construct} presents our main construction along with some basic properties. In Section~\ref{sec:localexp}, we compute the local expansion of the construction, from which a result of Alev and Lau~\cite{alev_improved_2020} implies rapid mixing of high-order walks. Section~\ref{sec:conc} discusses potential generalizations and limitations.

\section{Background and preliminaries}
\label{sec:prelim}
This section provides basic definitions pertaining to simplicial complexes and high-dimensional expanders, as well as relevant past results.
\begin{definition}
  A \textbf{simplicial complex} $X$ on $n$ vertices is a subset $X\subset 2^{[n]}$ such that if $\sigma\in X$, then all subsets of $\sigma$ also belong to $X$. Let $X(k)=\{\sigma\in X:|\sigma|=k+1\}$, and let the \textbf{$k$-skeleton} of $X$ refer to the simplicial complex $\bigcup_{\ell\leq k}X(k)$. The elements of $X(k)$ will be referred to as \textbf{$k$-dimensional faces}. The \textbf{dimension} of a simplicial complex is the maximum dimension of any of its faces, and if each face is contained in a face of maximal dimension, then the complex is \textbf{pure}. A \textbf{balanced weight function} $m:X\rightarrow\bR_+$ on a pure simplicial complex $X$ is a function such that for every $-1\leq k<\dim(X)$ and every $\sigma\in X(k)$, $$m(\sigma)=\sum_{\tau:\sigma\subset\tau\in X(k+1)}m(\tau).$$ 
\end{definition}
The 1-skeleton of a simplicial complex $X$ with balanced weight function $m$ is the undirected weighted graph $(X(0),X(1),m)$ that has vertices $[n]=X(0)$, edges $X(1)$, and edge weights $m(e)$ for $e\in X(1)$. The weighted degree of a vertex $x$ in this graph is given by the weight $m(x)$.

This paper will restrict attention to pure weighted simplicial complexes with balanced weight functions. In this case, the faces and weight function may be defined only on faces of maximal dimension, then propagated downwards, and the following useful formula applies.

\begin{lemma}[\cite{oppenheim_local_2018}]
  \label{lem:weightdef}
  For every $H$-dimensional simplicial complex $X$, every $-1\leq k\leq H$, and every $\sigma\in X(k)$, $$m(\sigma)=(H-k)!\sum_{\tau:\sigma\subset\tau\in X(H)}m(\tau).$$
\end{lemma}

Just as ordinary expander graph families are specified to have bounded degree, we are interested in families of simplicial complexes satisfying an analagous notion:
\begin{definition}
  A family $\cX$ of simplicial complexes has \textbf{bounded degree} if there exists some constant $\overline{d}$ such that for every $X\in\cX$ and every vertex $x\in X(0)$, there are at most $\overline{d}$ faces in $X$ that contain $x$.
\end{definition}

The local properties of a simplicial complex are captured by the links of faces, defined below.
\begin{definition}
  For a simplicial complex $X$ with weight function $m$, the \textbf{link} $X_\sigma$ of any $\sigma\in X$ is the simplicial complex defined by $X_\sigma=\{\tau\setminus\sigma:\sigma\subset\tau\in X\}$ with weight function $m_\sigma(\tau\setminus\sigma)=m(\tau)$.
\end{definition}
A common theme in the study of high-dimensional expanders is the ``local-to-global'' paradigm, which uses bounds on the expansion of the 1-skeletons of links to prove global expansion properties. To state such local-to-global results, it is first necessary to define graph expansion.

\begin{definition}
  For a graph $G$, the adjacency matrix is denoted $M_G$, the diagonal degree matrix is denoted $D_G$, and the random walk matrix is denoted $W_G=M_GD_G^{-1}$. The eigenvalues of the random walk matrix are denoted from greatest to least by $\omega_i(G)=\omega_i(W_G)$. The \textbf{expansion}, or \textbf{spectral gap}, of $G$ is the quantity $\nu_2(G)=\nu_2(W_G)=1-\omega_2(G)$.
\end{definition}
For any $n$-vertex graph $G$, all eigenvalues $\omega_i(W_G)$ of the random walk matrix lie in $[-1,1]$, and $\omega_1(W_G)=1$, so $0\leq\nu_2(G)\leq 1+\frac{1}{n-1}$ because $\sum_i\omega_i(W_G)=\Tr(W_G)\geq 0$. Graphs with spectral gaps closer to $1$ are considered ``better'' expanders.

We now introduce a notion of expansion for simplicial complexes.
\begin{definition}
  For $-1\leq k\leq H-2$, the \textbf{$k$-dimensional local expansion} of a simplicial complex $X$ with weight function $m$ is the value $$\nu^{(k)}(X)=\min_{\sigma\in X(k)}\nu_2(X_\sigma(0),X_\sigma(1),m_\sigma).$$ The \textbf{local expansion} of $X$ is the minimum of the $k$-dimensional local expansion over all $k\geq 0$, while the \textbf{global expansion} equals the $(-1)$-dimensional local expansion.
\end{definition}
That is, the $k$-dimensional local expansion refers to the lowest expansion of the 1-skeleton of the link of any $k$-dimensional face. Note that the terminology $k$-dimensional local expansion is nonstandard, but will be useful here. The following result shows that good local expansion in higher dimensions implies good local expansion in lower dimensions.
\begin{proposition}[\cite{oppenheim_local_2018}]
  \label{prop:localdown}
  Let $X$ be a simplicial complex in which all links of dimension $\geq 1$ are connected. Then for every $0\leq k\leq\dim(X)-2$, $$\nu^{(k-1)}(X)\geq 2-\frac{1}{\nu^{(k)}(X)}.$$
\end{proposition}
In particular, Proposition~\ref{prop:localdown} implies that if $\nu^{(k)}(X)\geq 1-\epsilon$, then $\nu^{(k-1)}(X)\geq 1-O(\epsilon)$.

The definition below presents the high-order random walks on simplicial complexes.
\begin{definition}
  Fix a pure $H$-dimensional simplicial complex $X$. For $-1\leq k\leq H-1$, define the \textbf{up-step random walk operator} $\Wu_k\in\bR^{X(k+1)\times X(k)}$ so that for $\sigma\in X(k),\tau\in X(k+1)$,
  \begin{align*}
    \Wu_k(\tau,\sigma) &= \begin{cases}
      \frac{m(\tau)}{m(\sigma)},&\sigma\subset\tau\in X(k+1)\\
      0,&\text{otherwise}
    \end{cases}
  \end{align*}
  For $0\leq k\leq H$, define the \textbf{down-step random walk operator} $\Wd_k\in\bR^{X(k-1)\times X(k)}$ so that for $\sigma\in X(k),\tau\in X(k-1)$,
  \begin{align*}
    \Wd_k(\tau,\sigma) &= \begin{cases}
      \frac{1}{k+1},&\sigma\supset\tau\in X(k-1) \\
      0,&\text{otherwise}.
    \end{cases}
  \end{align*}
  Define the \textbf{up-down} and \textbf{down-up random walk operators} by $\Wud_k=\Wd_{k+1}\circ\Wu_k$ and $\Wdu_k=\Wu_{k-1}\circ\Wd_k$ respectively.
\end{definition}
For context with this definition, consider a 1-dimensional simplicial complex $X$, which may be viewed as a weighted, undirected graph. Then the up-step operator $\Wu_0$ moves from a vertex to an adjacent edge with probability proportial to its weight, while the down-step operator $\Wd_1$ moves from an edge to either of its vertices with probability $\frac12$. Thus $\Wud_0=\Wd_1\circ\Wu_0$ is the ordinary graph lazy random walk operator, with stationary probability $\frac12$.

The nonzero elements of the spectra of $\Wud_k$ and of $\Wdu_{k+1}$ are identical. Therefore when studying the expansion of these operators, we restrict attention without loss of generality to $\Wud_k$.

The spectral gap $\nu_2(\Wud_k)$ gives a measure of high-dimensional expansion. Local expansion, defined above, provides another notion of high-dimensional expansion. The following result, which follows the ``local-to-global'' paradigm, shows that these two notions are closely related.
\begin{theorem}[\cite{alev_improved_2020}]
  \label{thm:localtomix}
  Let $X$ be an $H$-dimensional simplicial complex, and let $W_k^{\ud}$ refer to the up-down walk operator on $X$. Then for every $0\leq k\leq H-1$, $$\nu_2(W_k^{\ud})\geq\frac{1}{k+2}\prod_{j=-1}^{k-1}\nu^{(j)}(X).$$
\end{theorem}
Thus if a simplicial complex has good local expansion, then its high-order walks have large spectral gaps. We apply this result to show our main high-dimensional expansion bound.

The bound in Theorem~\ref{thm:localtomix} is nearly tight for good local expanders:
\begin{proposition}[\cite{alev_improved_2020}]
  \label{prop:gapupperbound}
  Let $X$ be an $H$-dimensional simplicial complex with at least $2(H+1)$ vertices, and let $\Wud_k$ refer to the up-down walk operator on $X$. Then for every $0\leq k\leq H-1$, $$\nu_2(\Wud_k)\leq\frac{2}{k+2}.$$
\end{proposition}

The main purpose of this paper is to present a construction of high-dimensional expander families, defined below.
\begin{definition}
  \label{def:expfam}
  A family $\cX$ of $H$-dimensional simplicial complexes is an \textbf{$H$-dimensional expander family} if the following conditions hold:
  \begin{enumerate}
  \item For every $n\in\bN$, there exists some $X\in\cX$ with $|X(0)|\geq n$.
  \item $\cX$ has bounded degree.
  \item\label{it:exp} For every $0\leq k\leq H-1$, there exists some $\underline{\nu}_{\cX,k}>0$ such that for every $X\in X$, the $k$-dimensional up-down walk operator $\Wud_k$ on $X$ satisfies $\nu_2(\Wud_k)\geq\underline{\nu}_{\cX,k}$.
  \end{enumerate}
\end{definition}
We are interested in constructing high-dimensional expander families $\cX$ of all dimensions $H$ with the up-down walk spectral gaps $\underline{\nu}_{\cX,k}$ close to the upper bound in Proposition~\ref{prop:gapupperbound}. By Theorem~\ref{thm:localtomix}, item~\ref{it:exp} in Definition~\ref{def:expfam} is implied by a uniform lower bound on the local expansion $\nu^{(k)}(X)$ of all $X\in\cX$ for every $-1\leq k\leq H-1$. We use this fact in the analysis of our construction.

\section{Construction}
\label{sec:construct}
The following definition introduces the simplicial complex $Z$ with weight function $m$ considered in this paper. The construction takes as input an $n$-vertex graph $G$, which will typically be chosen from a family of expanders, as well as a dimension $H$ and a parameter $s\geq H+1$, the latter of which will typically be taken as $s=2H$. The parameters $H$ and $s$ will typically be treated as fixed values, while $G$ and $n$ vary, so that the construction provides a family of $H$-dimensional simplicial complexes.
\begin{definition}
  \label{def:construct}
Let $G=(V(G),E(G),w_G)$ be any connected undirected graph on $n$ vertices with no self-loops. For positive integers $H$ and $s$, define the $H$-dimensional simplicial complex $Z$ with vertex set $V(G)\times[s]$ such that
\begin{align*}
  Z(H) = \{\{(v_1,b_1),\dots,(v_{H+1},b_{H+1})\}\subset V(G)\times[s]:
  \;&\exists \{u,v\}\in E(G)\text{ s.t. }\{v_1,\dots,v_{H+1}\}=\{u,v\}, \\
    &b_1,\dots,b_{H+1}\text{ are all distinct}\}.
\end{align*}
Define a weight function $m$ on $Z$ so that if $\sigma=\{(v_1,b_1),\dots,(v_{H+1},b_{H+1})\}\in Z(H)$ is such that $|\{i:v_i=u\}|=j$ and $|\{i:v_i=v\}|=H+1-j$ for some edge $\{u,v\}\in E(G)$, then let $$m(\sigma)=\frac{w_G(\{u,v\})}{{H-1\choose j-1}}.$$
\end{definition}
In words, the $H$-dimensional faces of $Z$ are those sets of the form $\{(v_1,b_1),\dots,(v_{H+1},b_{H+1})\}$ such that all $v_i$ are contained in a single edge of $G$, and such that all $b_i$ are distinct. If the above definition is modified so that the condition $\{v_1,\dots,v_{H+1}\}=\{u,v\}$ is replaced with $\{v_1,\dots,v_{H+1}\}\subset\{u,v\}$ and so that $m(\sigma)=1$ for all $H$-dimensional faces $\sigma$, then the resulting simplicial complex $Q$ is exactly the construction of Liu et al.~\cite{liu_high-dimensional_2020}. The difference between the weight functions of $Z$ and $Q$ lead to the key insights of this paper. Also note that here $G$ may be a weighted graph, whereas Liu et al.~\cite{liu_high-dimensional_2020} only considered complexes $Q$ derived from unweighted graphs.

If $G$ is unweighted so that $w_G(\{u,v\})=1$ for all $\{u,v\}\in E(G)$, then every $\sigma\in Z$ has $(H-1)!m(\sigma)\in\bN$. Thus $Z$ may be viewed as an unweighted simplicial complex, meaning that all $H$-dimensional faces have the same weight, but where multiple copies of faces are permitted.

The simplicial complex $Z$ may be viewed as a sort of high-dimensional tensor product of the graph $G$ and the $s$-vertex, $H$-dimensional complete complex $\cK_{[s]}$. In particular, the faces in $Z(H)$ are exactly those $(H+1)$-element subsets of $V(G)\times[s]$ for which the projection onto the first component gives an edge in $G$, and the projection onto the second component gives a face in $\cK_{[s]}(H)$. For comparison, an analagous property holds for the ordinary graph tensor product $G_1\otimes G_2$, in which the edges are given by those pairs of elements of $V(G_1)\times V(G_2)$ whose projection onto either component gives an edge in the respective graph $G_1$ or $G_2$.

Note that as with $Q$, the simplicial complex $Z$ has bounded degree with respect to~$n$ if $G$ has bounded degree. In particular, fix some values $H$ and $s$, and let $G$ be chosen from a family of bounded degree graphs. For every $\{u,v\}\in E(G)$ and $1\leq j\leq H$, by definition
\begin{equation}
  \label{eq:facesetsize}
  |\{\sigma=\{(v_1,b_1),\dots,(v_{H+1},b_{H+1})\}\in Z(H):|\{i:v_i=u\}|=j\}|={s\choose H+1}{H+1\choose j}.
\end{equation}
It follows by the definition of $m$ that both the maximum weight and the maximum number of faces containing any face in $Z$ are bounded by a constant. Furthermore, the number of faces in $Z$ grows as $O(n)$.

An additional consequence of Equation~(\ref{eq:facesetsize}) is that while the cardinality of the set on the left hand side, which equals the sum of the weights of the faces in this set under the weight function of $Q$, has an exponential dependence on $j$, the sum of the weights of the faces in this set under the weight function of $Z$ is equal to ${s\choose H+1}{H+1\choose j}/{H-1\choose j-1}$, which has only a quadratic dependence on $j$. This observation provides some initial intuition for the exponential speedup of high-order walks on $Z$ compared to $Q$.

From here on, the weight function $m$ and the operators $\Wu_k,\Wd_k,\Wud_k,\Wdu_k$ will always refer to $Z$, unless explicitly stated otherwise.

\subsection{Main result}
\label{sec:mainresult}
Our main result, shown in Section~\ref{sec:localexp}, is stated below.
\begin{theorem}[Restatement of Corollary~\ref{cor:localapp}]
  Let $\Wud_k$ be the up-down walk operator for the simplicial complex $Z$ of Definition~\ref{def:construct}. If $H\geq 2$, $s\geq 2H$, and $n\geq 4$, then for every $0\leq k\leq H-1$, $$\nu_2(W_k^{\ud})\geq \frac{\nu_2(G)}{(1+\log H)(k+2)(k+1)}.$$
\end{theorem}
In contrast to this quadratic dependence on $k$, Alev and Lau~\cite{alev_improved_2020} showed that the spectral gap of the $k$-dimensional up-down walk on $Q$ is at least $\frac{c\nu_2(G)}{(k+2)2^{k+2}}$, where $c=c(G,H)\in[1,2]$ depends on the structure of $G$. The discussion in Section~\ref{sec:relweights} below provides intuition for why this exponential dependence in $k$ arises for $Q$, and how the adjusted weights in $Z$ yield the improved quadratic dependence.

Theorem~\ref{sec:mainresult} implies that for any fixed $H$ with $s=2H$, if $G$ is chosen from a family of bounded degree expanders with spectral gap $\nu>0$, the resulting simplicial complexes $Z$ form a family of $H$-dimensional expanders with $k$-dimensional up-down walk spectral gap at least $\frac{\nu}{(1+\log H)(k+2)(k+1)}$. For comparison, an optimal $H$-dimensional expander family, as is given by simplicial complexes with optimal local expansion by Theorem~\ref{thm:localtomix}, achieves a spectral gap of $\Omega(\frac{1}{k})$ for the $k$-dimensional up-down walk.

\subsection{Decomposition into permutation-invariant subsets}
This section formalizes the intuitive notion that the construction of $Z$ treats elements of $[s]$ interchangeably, and introduces some notation to reflect this symmetry. For any $1\leq k\leq H+1$, the set of $(k-1)$-dimensional faces $Z(k-1)$ may be decomposed as follows. For any $\{u,v\}\in E(G)$ and $0\leq j\leq k$, let $$Z((j,k-j)_{(u,v)})=\{\{(v_1,b_1),\dots,(v_{k},b_{k})\}\in Z(H):|\{i:v_i=u\}|=j,|\{i:v_i=v\}|=k-j\}$$ be the set of all $(k-1)$-dimensional faces in $Z$ that contain $j$ vertices in $\{u\}\times[s]$ and $k-j$ vertices in $\{v\}\times[s]$. To remove redundancy when $j=k$, let $$Z((k)_u)=Z((k,0)_{(u,v)}).$$ Then by construction, $$Z(k-1)=\bigsqcup_{u\in V(G)}Z((k)_u)\sqcup\bigsqcup_{\{u,v\}\in E(G),1\leq j\leq k-1}Z((j,k-j)_{(u,v)}).$$ This decomposition simply partitions $Z$ into faces that differ only by permutations of $[s]$:
\begin{lemma}
  \label{lem:permact}
  For every $1\leq k\leq H+1$, there is a group action of $S_{[s]}=\{\text{permutations }\pi:[s]\rightarrow[s]\}$ on the set of faces of $Z$ given by
  \begin{equation*}
    \pi(\{(v_1,b_1),\dots,(v_k,b_k)\})=\{(v_1,\pi(b_1)),\dots,(v_k,\pi(b_k))\}.
  \end{equation*}
  The orbits of this action are exactly the sets $Z((j,k-j)_{(u,v)})$. The action preserves weights, that is, $m\circ\pi=m$.
\end{lemma}
\begin{proof}
  The construction of $Z$ directly implies that for all $\pi$, if $\sigma\in Z$ then $\pi(\sigma)\in Z$, so the group action on $Z$ is well defined. Similarly, for all $\pi$, the definition of $Z((j,k-j)_{(u,v)})$ directly implies that $\pi(Z((j,k-j)_{(u,v)}))=Z((j,k-j)_{(u,v)})$. For any $\sigma=\{(v_1,b_1),\dots,(v_k,b_k)\},\sigma'=\{(v_1',b_1'),\dots,(v_k',b_k')\}\in Z((j,k-j)_{(u,v)})$, if $\pi\in S_{[s]}$ is any permutation such that $\pi(\{b_i:v_i=u\})=\{b_i':v_i'=u\}$ and $\pi(\{b_i:v_i=v\})=\{b_i':v_i'=v\}$, then $\pi(\sigma)=\sigma'$. Thus $Z((j,k-j)_{(u,v)})$ is the orbit of $\sigma$ under the group action. For any $\pi$, to verify that $m\circ\pi=m$, note that by definition $m$ is constant over all values of $Z((j,H+1-j)_{(u,v)})$ for any given $0\leq j\leq H+1$ and $\{u,v\}\in E(G)$. Thus for all $\sigma\in Z(H)$, the characterization of the orbits above implies that $m(\pi(\sigma))=m(\sigma)$. This equality then extends to $\sigma$ of any dimension by Lemma~\ref{lem:weightdef}.
\end{proof}

Loosely speaking, Lemma~\ref{lem:permact} says that elements of $Z((j,k-j)_{(u,v)})$ may be treated interchangably, which in particular permits the following definition.
\begin{definition}
  For all $1\leq k\leq H+1$, $0\leq j\leq k$, and $\{u,v\}\in E(G)$, choose any $\sigma\in Z((j,k-j)_{(u,v)})$ and define $w_{(u,v)}^{(j,k-j)}=m(\sigma)$. This definition does not depend on the choice of $\sigma\in Z((j,k-j)_{(u,v)})$ by Lemma~\ref{lem:permact}. To avoid redundancy, also define $w_{(u)}^{(k)}=w_{(u,v)}^{(k,0)}$.
\end{definition}
Note that $m$ is defined by letting $w_{(u,v)}^{(j,H+1-j)}=w_G(\{u,v\})/{H-1\choose j-1}$. A basic property of these weights is that for any $1\leq j\leq k-1$, the ratio $w_{(u,v)}^{(j,k-j)}/w_G(\{u,v\})$ is independent of the choice of edge $\{u,v\}\in E(G)$, as is shown below.
\begin{lemma}
  \label{lem:weightratio}
  For all $2\leq k\leq H+1$, $1\leq j\leq k-1$, and $\{u,v\}\in E(G)$,
  \begin{equation*}
    \frac{w_{(u,v)}^{(j,k-j)}}{w_G(\{u,v\})} = (H+1-k)!\sum_{\ell=0}^{H+1-k}{s-k\choose \ell}{s-k-\ell\choose H+1-k-\ell}\cdot\frac{1}{{H-1\choose j+\ell-1}}.
  \end{equation*}
\end{lemma}
\begin{proof}
  For any $\sigma\in Z((j,k-j)_{(u,v)})$, any $H$-dimensional face $\tau\supset\sigma$ must satisfy $\tau\in Z((j+\ell,H+1-j-\ell)_{(u,v)})$ for some $0\leq \ell\leq H+1-k$. Therefore by Lemma~\ref{lem:weightdef},
  \begin{align*}
    \frac{m(\sigma)}{w_G(\{u,v\})}
    &= \frac{(H+1-k)!\sum_{\ell=0}^{H+1-k}\sum_{\sigma\subset\tau\in Z((j+\ell,H+1-j-\ell)_{(u,v)})}m(\tau)}{w_G(\{u,v\})} \\
    &= (H+1-k)!\sum_{\ell=0}^{H+1-k}{s-k\choose \ell}{s-k-\ell\choose H+1-k-\ell}\cdot\frac{1}{{H-1\choose j+\ell-1}},
  \end{align*}
  where the final equality holds because there are exactly ${s-k\choose \ell}{s-k-\ell\choose H+1-k-\ell}$ elements $\tau\in Z((j+\ell,H+1-j-\ell)_{(u,v)})$ such that $\tau\supset\sigma$, and for each one $m(\tau)=w_{(u,v)}^{(j+\ell,H+1-j-\ell)}=w_G(\{u,v\})/{H-1\choose j+\ell-1}$ by definition.
\end{proof}

\subsection{Relative weights of overlapping faces}
\label{sec:relweights}
The proposition below determines the relative weights of faces of $Z$ that intersect at all but one of their vertices. This result is used in Section~\ref{sec:localexp} to determine the local expansion of $Z$.

\begin{proposition}
  \label{prop:weights}
  For all $1\leq k\leq H$, $1\leq j\leq k-1$, and $\{u,v\}\in E(G)$, it holds that $$\frac{w_{(u,v)}^{(j+1,k-j)}}{w_{(u,v)}^{(j,k-j+1)}}=\frac{j}{k-j}.$$ Furthermore, $$\frac{w_{(u)}^{(k+1)}}{\sum_{v\in N(u)}w_{(u,v)}^{(k,1)}}=k\sum_{i=k+1}^H\frac{1}{i}.$$
\end{proposition}
\begin{proof}
  Both statements are shown using induction. For the first equality, the base case $k=H$ follows by the definition of $w_{(u,v)}^{(j,H+1-j)}=w_G(\{u,v\})/{H-1\choose j-1}$, so that $$\frac{w_{(u,v)}^{(j+1,H-j)}}{w_{(u,v)}^{(j,H-j+1)}}=\frac{{H-1\choose j-1}}{{H-1\choose j}}=\frac{j}{H-j}.$$ For the inductive step, assume for some $1\leq k\leq H-1$ that it holds for all $1\leq j\leq k$ that $w_{(u,v)}^{(j+1,k+1-j)}/w_{(u,v)}^{(j,k+2-j)}=j/(k+1-j).$ For any $1\leq j\leq k-1$ and any $\sigma\in Z((j+1,k-j)_{(u,v)})$, by definition any $(k+1)$-dimensional face $\tau\supset\sigma$ is obtained from $\sigma$ by adding either a vertex in $\{u\}\times([s]\setminus\Pi_{[s]}(\sigma))$ or in $\{v\}\times([s]\setminus\Pi_{[s]}(\sigma))$. Therefore
  \begin{align}
    \label{eq:weightup}
    \begin{split}
      w_{(u,v)}^{(j+1,k-j)} = m(\sigma)
      &= \sum_{\sigma\subset\tau\in Z(k+1)}m(\tau) \\
      &= \sum_{\sigma\subset\tau\in Z((j+2,k-j)_{(u,v)})}m(\tau)+\sum_{\sigma\subset\tau\in Z((j+1,k-j+1)_{(u,v)})}m(\tau) \\
      &= (s-k-1)(w_{(u,v)}^{(j+2,k-j)}+w_{(u,v)}^{(j+1,k-j+1)}).
    \end{split}
  \end{align}
  Applying the exact same reasoning to a face $\sigma'\in Z((j,k-j+1)_{(u,v)}$ gives that $$w_{(u,v)}^{(j,k-j+1)}=m(\sigma')=(s-k-1)(w_{(u,v)}^{(j+1,k-j+1)}+w_{(u,v)}^{(j,k-j+2)}).$$ Therefore
  \begin{align*}
    \frac{w_{(u,v)}^{(j+1,k-j)}}{w_{(u,v)}^{(j,k-j+1)}}
    &= \frac{(s-k-1)(w_{(u,v)}^{(j+2,k-j)}+w_{(u,v)}^{(j+1,k-j+1)})}{(s-k-1)(w_{(u,v)}^{(j+1,k-j+1)}+w_{(u,v)}^{(j,k-j+2)})} \\
    &= \frac{w_{(u,v)}^{(j+2,k-j)}/w_{(u,v)}^{(j+1,k-j+1)}+1}{1+w_{(u,v)}^{(j,k-j+2)}/w_{(u,v)}^{(j+1,k-j+1)}} \\
    &= \frac{(j+1)/(k-j)+1}{1+(k-j+1)/j} \\
    &= \frac{j}{k-j},
  \end{align*}
  completing the inductive step; note that the third equality above holds by the inductive hypothesis.

  To show the second equality in the proposition statement, first note that the base case $k=H$ holds immediately as $w_{(u)}^{(H+1)}=0$ because the definition of the complex $Z$ does not include, or equivalently assigns zero weight, to faces in $Z((H+1)_{(u)})$. For the inductive step, assume that for some $1\leq k\leq H-1$ it holds that $w_{(u)}^{(k+2)}/\sum_{v\in N(u)}w_{(u,v)}^{(k+1,1)}=(k+1)\sum_{i=k+2}^H\frac{1}{i}.$ For any $\sigma\in Z((k+1)_{(u)})$, by definition any $(k+1)$-dimensional face $\tau\supset\sigma$ is obtained from $\sigma$ by adding either a vertex in $\{u\}\times([s]\setminus\Pi_{[s]}(\sigma))$ or in $\{v\}\times([s]\setminus\Pi_{[s]}(\sigma))$ for some $v\in N(u)$. Therefore
  \begin{align*}
    w_{(u)}^{(k+1)} = m(\sigma)
    &= \sum_{\sigma\subset\tau\in Z(k+1)}m(\tau) \\
    &= \sum_{\sigma\subset\tau\in Z((k+2)_{(u)})}m(\tau) + \sum_{v\in N(u)}\sum_{\sigma\subset\tau\in Z((k+1,1)_{(u,v)})}m(\tau) \\
    &= (s-k-1)\left(w_{(u)}^{(k+2)}+\sum_{v\in N(u)}w_{(u,v)}^{(k+1,1)}\right).
  \end{align*}
  Applying~(\ref{eq:weightup}) with $j=k-1$ gives that $$\sum_{v\in N(u)}w_{(u,v)}^{(k,1)}=(s-k-1)\left(\sum_{v\in N(u)}w_{(u,v)}^{(k+1,1)}+\sum_{v\in N(u)}w_{(u,v)}^{(k,2)}\right).$$ Therefore
  \begin{align*}
    \frac{w_{(u)}^{(k+1)}}{\sum_{v\in N(u)}w_{(u,v)}^{(k,1)}}
    &= \frac{(s-k-1)\left(w_{(u)}^{(k+2)}+\sum_{v\in N(u)}w_{(u,v)}^{(k+1,1)}\right)}{(s-k-1)\left(\sum_{v\in N(u)}w_{(u,v)}^{(k+1,1)}+\sum_{v\in N(u)}w_{(u,v)}^{(k,2)}\right)} \\
    &= \frac{w_{(u)}^{(k+2)}/\left(\sum_{v\in N(u)}w_{(u,v)}^{(k+1,1)}\right)+1}{1+\left(\sum_{v\in N(u)}w_{(u,v)}^{(k,2)}\right)/\left(\sum_{v\in N(u)}w_{(u,v)}^{(k+1,1)}\right)} \\
    &= \frac{(k+1)\sum_{i=k+2}^H1/i+1}{1+1/k} \\
    &= k\sum_{i=k+1}^H\frac{1}{i},
  \end{align*}
  completing the inductive step; note that the third equality above holds by the inductive hypothesis, and because $w_{(u,v)}^{(k,2)}=w_{(u,v)}^{(k+1,1)}/k$ for all $v\in N(u)$ as was shown above.
\end{proof}

Proposition~\ref{prop:weights} provides the key insight for understanding why the spectral gap of the up-down walk on $Z$ has a quadratic dependence in $k$, whereas that of $Q$ has an exponential dependence. For some $2\leq k\leq H$, $1\leq j\leq k-1$, $\{u,v\}\in E(G)$, consider an element $\sigma\in Z((j,k-j)_{(u,v)})$. Let $\sigma'\sim\Wud_{k-1}\1_\sigma$ be the random variable for the face obtained by taking one step in the up-down walk starting at $\sigma$. Then $\sigma'\in Z((j',k-j')_{(u,v)})$ for some $j'\in\{j+1,j,j-1\}$. Let $\Pi_{[s]}(\sigma)$ denote the subset of $[s]$ obtained by projecting the elements of $\sigma$ to their second components. If $j'=j+1$, then the up step must add some vertex in $\{u\}\times([s]\setminus\Pi_{[s]}(\sigma))$ and the down step must remove some vertex in $(\{v\}\times[s])\cap\sigma$, while if $j'=j-1$ then the up step must add some vertex in $\{v\}\times([s]\setminus\Pi_{[s]}(\sigma))$ and the down step must remove some vertex in $(\{u\}\times[s])\cap\sigma$. Thus by the definition of the up- and down-step transition probabilities,
\begin{align}
  \label{eq:transprobs}
  \begin{split}
    \Pr[j'=j+1] &= \frac{w_{(u,v)}^{(j+1,k-j)}}{w_{(u,v)}^{(j+1,k-j)}+w_{(u,v)}^{(j,k-j+1)}}\cdot\frac{k-j}{k+1} \\
    \Pr[j'=j-1] &= \frac{w_{(u,v)}^{(j,k-j+1)}}{w_{(u,v)}^{(j+1,k-j)}+w_{(u,v)}^{(j,k-j+1)}}\cdot\frac{j}{k+1}.
  \end{split}
\end{align}
For the construction $Q$ of Liu et al.~\cite{liu_high-dimensional_2020}, these same expressions hold, but $w_{(u,v)}^{(j+1,k-j)}=w_{(u,v)}^{(j,k-j+1)}$, so that when $j$ is close to $1$ or close to $k$, the transition probabilities in~(\ref{eq:transprobs}) are heavily skewed to push $j'$ in the direction of $k/2$. It is this property that results in an exponential dependence on $k$ in the $k$th order up-down walk on $Q$; the up-down walk becomes ``trapped'' in the set of faces contained in $\{u,v\}\times[s]$, with the transition probabilities pushing away from the ``exit routes'' $Z((k)_{(u)})$ and $Z((k)_{(v)})$.

To understand why the weight function $m$ on $Z$ resolves this issue, observe that for $Z$, Proposition~\ref{prop:weights} implies that both probabilities in~(\ref{eq:transprobs}) equal $\frac{j(k-j)}{k(k+1)}$. Therefore the events $j'=j+1$ and $j'=j-1$ are equally likely. Thus the up-down walk moves across the sets $Z((j,k-j)_{(u,v)})$ for $1\leq j\leq k-1$ as a lazy random walk on an unweighted, undirected, $(k-1)$-vertex path. The mixing time for such a walk grows quadratically in $k$, thereby providing intuition for the quadratic dependence in $k$ for the mixing time of $\Wud_{k-1}$.

The intuition described above can be formalized to bound the mixing time of the high-order walks on $Z$. However, the following section takes a different approach by computing the local expansion of $Z$, which leads to tighter bounds on $\nu_2(\Wud_k)$.

\section{Local and global expansion}
\label{sec:localexp}
This section analyzes the local and global expansion of $Z$, which is then used to bound the mixing time of the up-down random walk using Theorem~\ref{thm:localtomix}.

\begin{theorem}
  \label{thm:localexp}
  If $H\geq 2$, $s\geq 2H$, and $n\geq 4$, then for every $0\leq k\leq H-2$, $$\nu^{(k)}(Z)=\frac{k+1}{k+2}.$$ Furthermore, $$\nu^{(-1)}(Z)=\frac{\nu_2(G)}{\sum_{\ell=1}^H1/\ell}\geq\frac{\nu_2(G)}{1+\log H}.$$
\end{theorem}
Note that the local expansion $\nu^{(k)}(Z)$ for $k\geq 0$ does not depend on $G$. This property stems from the fact that the structure of any given link in $Z$ depends only on the local structure of $G$, that is, on the weights of edges adjacent to a single vertex.

For comparison, the construction $Q$ of Liu et al.~\cite{liu_high-dimensional_2020} has local expansion $\nu^{(k)}(Q)=\frac12$ in each dimension $k\geq 0$, and has global expansion $\nu^{(-1)}(Q)$ approaching $\frac12\nu_2(G)$ as $H$ grows large. It was posed as an open question in Liu et al.~\cite{liu_high-dimensional_2020} whether $\frac12$ is a natural barrier for local expansion in graph-product-based constructions. Theorem~\ref{thm:localexp} gives a partial answer to this question, as although $Z$ sacrifices a factor of $O(\log H)$ in global expansion compared to $Q$, for all $k\geq 1$ the local expansion $\nu^{(k)}(Z)=\frac{k+1}{k+2}$ is an improvement on $\nu^{(k)}(Q)=\frac12$. Section~\ref{sec:conc} provides a discussion suggesting that further improvements in local expansion are not attainable with a similar graph-product-based construction.

The result below applies Theorem~\ref{thm:localtomix}, shown by Alev and Lau~\cite{alev_improved_2020}, to show that the improvement in local expansion of $Z$ compared to $Q$ for $k\geq 1$ results in an exponential improvement with respect to $k$ of the spectral gap of the $k$th order up-down walk.
\begin{corollary}
  \label{cor:localapp}
  Let $\Wud_k$ be the up-down walk operator for the simplicial complex $Z$. If $H\geq 2$, $s\geq 2H$, and $n\geq 4$, then for all $0\leq k\leq H-1$, $$\nu_2(W_k^{\ud}) \geq \frac{\nu_2(G)}{(\sum_{\ell=1}^H1/\ell)(k+2)(k+1)} \geq \frac{\nu_2(G)}{(1+\log H)(k+2)(k+1)}.$$
\end{corollary}
\begin{proof}
  Applying Theorem~\ref{thm:localtomix} with Theorem~\ref{thm:localexp} gives that
  \begin{align*}
    \nu_2(W_k^{\ud})
    &\geq \frac{1}{k+2}\cdot\frac{\nu_2(G)}{\sum_{\ell=1}^H1/\ell}\prod_{j=0}^{k-1}\frac{j+1}{j+2}
      = \frac{\nu_2(G)}{(\sum_{\ell=1}^H1/\ell)(k+2)(k+1)}.
  \end{align*}
\end{proof}
Because by definition $W_k^{\ud}$ has self loop probabilities of $1/(k+2)$, for all $i$ it holds that $\omega_i(W_k^{\ud})\geq-k/(k+2)$. Therefore assuming that $G$ is chosen from a family of graphs with bounded ratio of maximum degree to minimum degree, then the mixing time of the random walk $W_k^{\ud}$ grows as $O(\frac{k^2(\log H)(\log n)}{\nu_2(G)})$.

In contrast, the spectral gap of the $k$th order random walk on the construction $Q$ of Liu et al.~\cite{liu_high-dimensional_2020} was shown in Alev and Lau~\cite{alev_improved_2020} to grow as $\Omega(\frac{\nu_2(G)}{k2^k})$, for a mixing time of $O(\frac{k2^k\log n}{\nu_2(G)})$.

\subsection{Proof of Theorem~\ref{thm:localexp}}
To prove Theorem~\ref{thm:localexp}, we first compute the expansion of the 1-skeleton of every link in $Z$ in the following lemmas.
\begin{lemma}
  \label{lem:inlocal}
  For every $2\leq k\leq H-1$, $1\leq j\leq k-1$, and $\{u,v\}\in E(G)$, every face $\sigma\in Z((j,k-j)_{(u,v)})$ satisfies $$\omega_2(Z_\sigma(0),Z_\sigma(1),m_\sigma)=\frac{1}{k+1}.$$
\end{lemma}
\begin{proof}
  Following the method of Liu et al.~\cite{liu_high-dimensional_2020}, the proof will proceed by identifying the 1-skeleton of each link of $Z$ with a tensor product of two other graphs, whose spectra can be analyzed independently. For a face $\sigma\in Z((j,k-j)_{(u,v)})$, the link $Z_\sigma$ by definition has vertex set $Z_\sigma(0)=\{u,v\}\times([s]\setminus\Pi_{[s]}(\sigma))$ and edge set $$Z_\sigma(1)=\{\tau\setminus\sigma:\sigma\subset\tau\in Z(k+1)\}=\{\{(v_1,b_1),(v_2,b_2)\}\subset Z_\sigma(0):b_1\neq b_2\}.$$ For such an edge $\{(v_1,b_1),(v_2,b_2)\}$, let $\tau=\sigma\cup\{(v_1,b_1),(v_2,b_2)\}$, so that the edge's weight $m(\tau)$ may be one of three possible values:
  \begin{itemize}
  \item If $v_1=v_2=u$, then $\tau\in Z((j+2,k-j)_{(u,v)})$, so $m(\tau)=w_{(u,v)}^{(j+2,k-j)}$.
  \item If $v_1=u,v_2=v$, then $\tau\in Z((j+1,k-j+1)_{(u,v)})$, so $m(\tau)=w_{(u,v)}^{(j+1,k-j+1)}$.
  \item If $v_1=v_2=v$, then $\tau\in Z((j,k-j+2)_{(u,v)})$, so $m(\tau)=w_{(u,v)}^{(j,k-j+2)}$.
  \end{itemize}
  Therefore letting $P$ be the 2-vertex graph with adjacency matrix $$M_P=\begin{pmatrix}w_{(u,v)}^{(j+2,k-j)}&w_{(u,v)}^{(j+1,k-j+1)}\\w_{(u,v)}^{(j+1,k-j+1)}&w_{(u,v)}^{(j,k-j+2)}\end{pmatrix},$$ then the graph $(Z_\sigma(0),Z_\sigma(1),m_\sigma)$ described above is exactly given by $P\otimes K_{[s]\setminus\Pi_{[s]}(\sigma)}$, where $K_V$ denotes the complete graph without self-loops on vertex set $V$. By Proposition~\ref{prop:weights}, it holds that $w_{(u,v)}^{(j+2,k-j)}/w_{(u,v)}^{(j+1,k-j+1)}=(j+1)/(k-j)$ and $w_{(u,v)}^{(j+1,k-j+1)}/w_{(u,v)}^{(j,k-j+2)}=j/(k-j+1)$, so the random walk matrix of $P$ is given by $$W_P=\frac{1}{k+1}\begin{pmatrix}j+1&j\\k-j&k-j+1\end{pmatrix},$$ whose second eigenvalue is $1/(k+1)$, with eigenvector $(1,-1)^T$. Thus because $1$ is the only positive eigenvalue of $W_{K_{[s]\setminus\Pi_{[s]}(\sigma)}}$, it follows that the second eigenvalue of $W_P\otimes W_{K_{[s]\setminus\Pi_{[s]}(\sigma)}}$ is $1/(k+1)$, as desired.
\end{proof}

\begin{lemma}
  \label{lem:outlocal}
  If $s\geq 2H$, then for every $1\leq k\leq H-1$ and $u\in V(G)$, every face $\sigma\in Z((k)_{(u)})$ satisfies $$\omega_2(Z_\sigma(0),Z_\sigma(1),m_\sigma)=\frac{1}{k+1}.$$
\end{lemma}
\begin{proof}
  The proof will proceed similarly to that of Lemma~\ref{lem:inlocal}. Consider a face $\sigma\in Z((k)_{(u)})$, and let $N(u)=\{v_1,\dots,v_d\}$. The link $Z_\sigma$ then has vertex set $Z_\sigma(0)=(\{u\}\cup N(u))\times([s]\setminus\Pi_{[s]}(\sigma))$ and edge set $$Z_\sigma(1)=\{\tau\setminus\sigma:\sigma\subset\tau\in Z(k+1)\}=\{\{(v,b_1),(v',b_2)\}\subset Z_\sigma(0):b_1\neq b_2,|\{v,v'\}\setminus\{u\}|\leq 1\}.$$ For such an edge $\{(v,b_1),(v',b_2)\}$, let $\tau=\sigma\cup\{(v,b_1),(v',b_2)\}$, so that there are three possible cases for the edge's weight $m(\tau)$:
  \begin{itemize}
  \item If $v=v'=u$, then $\tau\in Z((k+2)_{(u)})$, so $m(\tau)=w_{(u)}^{(k+2)}$.
  \item If $v=u,v'=v_i$ for some $i$, then $\tau\in Z((k+1,1)_{(u,v_i)})$, so $m(\tau)=w_{(u,v_i)}^{(k+1,1)}$.
  \item If $v=v'=v_i$ for some $i$, then $\tau\in Z((k,2)_{(u,v_i)})$, so $m(\tau)=w_{(u,v_i)}^{(k,2)}$.
  \end{itemize}

  Therefore letting $S$ be the $(d+1)$-vertex star graph with adjacency matrix
  \begin{equation*}
    M_S = \begin{pmatrix}
      w_{(u)}^{(k+2)} & w_{(u,v_1)}^{(k+1,1)} & w_{(u,v_2)}^{(k+1,1)} & \cdots & w_{(u,v_d)}^{(k+1,1)} \\
      w_{(u,v_1)}^{(k+1,1)} & w_{(u,v_1)}^{(k,2)} & 0 & \cdots & 0 \\
      w_{(u,v_2)}^{(k+1,1)} & 0 & w_{(u,v_2)}^{(k,2)} & & \vdots \\
      \vdots & \vdots & & \ddots & 0 \\
      w_{(u,v_d)}^{(k+1,1)} & 0 & \hdots & 0 & w_{(u,v_d)}^{(k,2)}
    \end{pmatrix},
  \end{equation*}
  it follows that the graph $(Z_\sigma(0),Z_\sigma(1),m_\sigma)$ is exactly given by $S\otimes K_{[s]\setminus\Pi_{[s]}(\sigma)}$. Let $x=w_{(u)}^{(k+2)}+\sum_{i=1}^dw_{(u,v_i)}^{(k+1,1)}$, so that by Proposition~\ref{prop:weights}, the random walk matrix of $S$ is therefore
  \begin{equation*}
    W_S = \begin{pmatrix}
      w_{(u)}^{(k+2)}/x & k/(k+1) & k/(k+1) & \cdots & k/(k+1) \\
      w_{(u,v_1)}^{(k+1,1)}/x & 1/(k+1) & 0 & \cdots & 0 \\
      w_{(u,v_2)}^{(k+1,1)}/x & 0 & 1/(k+1) & & \vdots \\
      \vdots & \vdots & & \ddots & 0 \\
      w_{(u,v_d)}^{(k+1,1)}/x & 0 & \hdots & 0 & 1/(k+1)
    \end{pmatrix},
  \end{equation*}
  Let $\delta_i\in\bR^{d+1}$ denote the $i$th basis vector, so that by this expression for $W_S$, every vector in the \mbox{codimension-2} subspace $\text{span}\{\vec{1},\delta_1\}^\perp$ is an eigenvector with eigenvalue $1/(k+1)$. The final eigenvector in $\text{span}\{\vec{1}\}^\perp$ is then given by $(w_{(u)}^{(k+2)}-x,w_{(u,v_1)}^{(k+1,1)},\dots,w_{(u,v_d)}^{(k+1,1)})^T$, with eigenvalue $$\frac{1}{k+1}-\frac{x-w_{(u)}^{(k+2)}}{x}=\frac{1}{k+1}-\frac{1}{1+(k+1)\sum_{\ell=k+2}^H1/\ell},$$ where the equality above holds by Proposition~\ref{prop:weights}. Thus $\omega_2(S)=1/(k+1)$. Because $s\geq 2H$ and $k\leq H-1$, it follows that $|[s]\setminus\Pi_{[s]}(\sigma)|=s-k\geq H+1$. Therefore the eigenvalues of $W_{K_{[s]\setminus\Pi_{[s]}(\sigma)}}$ are $1$ and $-1/(s-k-1)$, with $0\leq 1/(s-k-1)\leq 1/H\leq 1/(k+1)$, so it follows that all eigenvalues of $W_{K_{[s]\setminus\Pi_{[s]}(\sigma)}}$ that do not equal $1$ must have absolute value at most $1/(k+1)$. Therefore it follows from $\omega_2(S)=1/(k+1)$ that $\omega_2(S\otimes K_{[s]\setminus\Pi_{[s]}(\sigma)})=1/(k+1)$, as desired.
\end{proof}

\begin{lemma}
  \label{lem:global}
  For $1\leq i\leq n$, let
  \begin{equation}
    \label{eq:lazyomega}
    \tilde{\omega_i}(G) = \frac{1}{\sum_{\ell=1}^H1/\ell}\omega_i(G)+\left(1-\frac{1}{\sum_{\ell=1}^H1/\ell}\right)
  \end{equation}
  denote the eigenvalues of a lazy random walk on $G$. Then $$\omega_2(Z(0),Z(1),m)=\max\{\tilde{\omega}_2(G),-\tilde{\omega}_n(G)/(s-1)\}.$$
\end{lemma}
\begin{proof}
  Consider any $\tau=\{(u,b_1),(v,b_2)\}\in Z(1)$. If $u=v$ then $\tau\in Z((2)_{(u)})$ so that $m(\tau)=w_{(u)}^{(2)}$, while if $u\neq v$ then $\tau\in Z((1,1)_{(u,v)})$ so that $m(\tau)=w_{(u,v)}^{(1,1)}$. Therefore define $\tilde{G}$ to be the undirected graph with $V(\tilde{G})=V(G)$, $E(\tilde{G})=E(G)\cup V(G)$, and with edge weight function $w_{\tilde{G}}(\cdot)$ given for $\{u,v\}\in E(G)$ by $w_{\tilde{G}}(\{u,v\})=w_{(u,v)}^{(1,1)}$ and $w_{\tilde{G}}(\{u\})=w_{(u)}^{(2)}$. Then the graph $(Z(0),Z(1),m)$ is exactly given by $\tilde{G}\otimes K_{[s]}$. Let $W_{\tilde{G}}$ denote the random walk matrix of $\tilde{G}$, and let $W_{\tilde{G}}'$ denote $W_{\tilde{G}}$ with all diagonal entries zeroed out, and let $W_{\tilde{G}}''$ denote $W_{\tilde{G}}$ with all non-diagonal entries zeroed out, so that $W_{\tilde{G}}=W_{\tilde{G}}'+W_{\tilde{G}}''$. Then for any $u\in V(G)$,
  \begin{align*}
    W_{\tilde{G}}''(u,u) &= \frac{w_{(u)}^{(2)}}{w_{(u)}^{(2)}+\sum_{v\in N(u)}w_{(u,v)}^{(1,1)}} = \frac{1}{1+\frac{\sum_{v\in N(u)}w_{(u,v)}^{(1,1)}}{w_{(u)}^{(2)}}} = \frac{\sum_{\ell=2}^H\frac{1}{\ell}}{\sum_{\ell=1}^H\frac{1}{\ell}},
  \end{align*}
  where the final equality holds by Proposition~\ref{prop:weights}. Therefore $W_{\tilde{G}}''=(\sum_{\ell=2}^H\frac{1}{\ell}/\sum_{\ell=1}^H\frac{1}{\ell})I$. Furthermore, for any $v\neq u$,
  \begin{align*}
    \left(\sum_{\ell=1}^H\frac{1}{\ell}\right)W_{\tilde{G}}'(v,u)
    &= \frac{w_{(u)}^{(2)}+\sum_{v'\in N(u)}w_{(u,v')}^{(1,1)}}{\sum_{v'\in N(u)}w_{(u,v')}^{(1,1)}} \cdot \frac{w_{(u,v)}^{(1,1)}}{w_{(u)}^{(2)}+\sum_{v'\in N(u)}w_{(u,v')}^{(1,1)}} \\
    &= \frac{w_{(u,v)}^{(1,1)}}{\sum_{v'\in N(u)}w_{(u,v')}^{(1,1)}} \\
    &= \frac{w_G(\{u,v\})}{\sum_{v'\in N(u)}w_G(\{u,v'\})} \\
    &= W_G(v,u),
  \end{align*}
  where the first equality above holds by Proposition~\ref{prop:weights}, and the third equality by Lemma~\ref{lem:weightratio}. Thus in summary,
  \begin{equation*}
    W_{\tilde{G}} = W_{\tilde{G}}'+W_{\tilde{G}}'' = \frac{1}{\sum_{\ell=1}^H1/\ell}W_G+\left(1-\frac{1}{\sum_{\ell=1}^H1/\ell}\right)I,
  \end{equation*}
  so $W_{\tilde{G}}$ is simply a lazy instance of the random walk $W_G$, and in particular for all $1\leq i\leq n$, $$\omega_i(\tilde{G})=\frac{1}{\sum_{\ell=1}^H1/\ell}\omega_i(G)+\left(1-\frac{1}{\sum_{\ell=1}^H1/\ell}\right)=\tilde{\omega}_i(G).$$ Because the eigenvalues of $K_{[s]}$ are $1$ and $-1/(s-1)$, it follows that $$\omega_2(\tilde{G}\otimes K_{[s]})=\max\{\omega_2(\tilde{G}),-\omega_n(\tilde{G})/(s-1)\},$$ as desired.
\end{proof}

\begin{proof}[Proof of Theorem~\ref{thm:localexp}]
  For every $0\leq k\leq H-2$, each $\sigma\in Z(k)$ by definition either lies in $Z((j,k+1-j)_{(u,v)})$ or in $Z((k+1)_{(u)})$ for some $1\leq j\leq k$ and $\{(u,v)\}\in E(G)$. Therefore Lemma~\ref{lem:inlocal} and Lemma~\ref{lem:outlocal} together imply that the link of every $\sigma\in Z(k)$ has expansion $\nu_2(Z_\sigma(0),Z_\sigma(1),m_\sigma)=\frac{k+1}{k+2}$, so $$\nu^{(k)}(Z)=\frac{k+1}{k+2}.$$

  For the global expansion statement, letting $\tilde{\omega}_i(G)$ be defined as in~(\ref{eq:lazyomega}), then by Lemma~\ref{lem:global}, $$\nu^{(-1)}(Z)=\nu_2(Z(0),Z(1),m)=1-\max\{\tilde{\omega}_2(G),-\tilde{\omega}_n(G)/(s-1)\}.$$ Now because $n\geq 4$, $H\geq 2$, and $s\geq 4$ by assumption, then $\omega_2(G)\geq-1/3$ and $\sum_{\ell=1}^H1/\ell\geq 3/2$, which implies that $\tilde{\omega_2}(G)\geq 1/9$ and $\tilde{\omega_n}(G)\geq-1/3$, and thus $\tilde{\omega_2}(G)\geq-\tilde{\omega}_n(G)/(s-1)$, so $$\nu^{(-1)}(Z)=1-\tilde{\omega_2}(G)=\frac{\nu_2(G)}{\sum_{\ell=1}^H1/\ell}\geq\frac{\nu_2(G)}{1+\log H}.$$
\end{proof}

\section{Discussion and future directions}
\label{sec:conc}
Given the constructions of ordinary expanders using graph products (e.g.~\cite{reingold_entropy_2002}), it seems natural to investigate simplicial complex product constructions that yield high-dimensional expanders. From this perspective, the construction $Q$ of Liu et al.~\cite{liu_high-dimensional_2020} is quite interesting, as it may be viewed as a form of tensor product between a graph $G$ and the complete simplicial complex on $s$ vertices. The principal drawback with $Q$ lies in the exponential dependence of the spectral gap $\Omega(\frac{\nu_2(G)}{k2^k})$ of the up-down walk on the dimension $k$. By reducing this dependence to quadratic with a spectral gap of $\Omega(\frac{\nu_2(G)}{k^2\log H})$, the construction $Z$ greatly improves the mixing time of high-dimensional up-down walks, while maintaing the product-based nature of the construction. However, the optimal spectral gap of the $k$th order up-down walk grows as $\Omega(\frac{1}{k})$, which is achieved by Ramanujan complexes and by the constructions based on coset geometries of Kaufman and Oppenheim~\cite{kaufman_high_2020} and Friedgut and Iluz~\cite{friedgut_hyper-regular_2020}. It is therefore natural to ask whether the construction $Z$ can be further optimized to reduce the quadratic dependence on $k$ to linear. Below, we suggest a negative answer to this question, by analyzing the implications for local expansion.

The determination of the optimal local expansion for any ``graph-product-based construction'' was posed as an open question by Liu et al.~\cite{liu_high-dimensional_2020}, although no formal definition for a graph-product-based construction was proposed. For concreteness, fix a dimension $H$, and let a graph-product-based construction be one such as $Z$ that takes as input a graph $G$, and outputs an $H$-dimensional simplicial complex with the same faces as $Q$, but with an arbitrary weight function.\footnote{The reasoning presented here also applies to more general constructions.} The following reasoning suggests that no such construction can improve upon the $k$-dimensional local expansion $\nu^{(k)}(Z)=\frac{k+1}{k+2}$ of $Z$. For if some graph-product-based construction $Z'$ were to satisfy $\nu^{(k)}(Z')>\frac{k+1+\epsilon}{k+2+\epsilon}$ for some $\epsilon>0$ when $G$ is chosen from a family of $d$-regular expanders, then inductively applying Proposition~\ref{prop:localdown} would imply that $\nu^{(j)}(Z')>\frac{j+1+\epsilon}{j+2+\epsilon}$ for all $j\leq k$, so that $\nu^{(-1)}(Z')>\frac{\epsilon}{1+\epsilon}$. But as observed in Section~\ref{sec:localexp}, the product-based structure ensures that the link of any face of dimension $\geq 0$ obtains its structure from the structure of the neighborhood of a single vertex in $G$, while the entire 1-skeleton of $Z'$ inherits the global structure of $G$. Thus if $G$ is instead chosen from a family of $d$-regular graphs with sufficiently poor expansion, then the 1-skeleton of $Z'$ will now inherit this poor expansion, so $\nu^{(-1)}(Z')<\frac{\epsilon}{1+\epsilon}$. But all $d$-regular graphs have the same local structure in the neighborhood of a vertex, so it will still hold that $\nu^{(k)}(Z')>\frac{k+1+\epsilon}{k+2+\epsilon}$, and therefore $\nu^{(-1)}(Z')>\frac{\epsilon}{1+\epsilon}$, a contradiction.

With $k$-dimensional local expansion $\nu^{(k)}=\frac{k+1}{k+2}$, the bound from Theorem~\ref{thm:localtomix} on the spectral gap of the $k$-dimensional up-down walk is at best $\frac{1}{(k+2)(k+1)}$. Indeed, a quadratic dependence on $k$ in mixing time also arose in the discussion in Section~\ref{sec:relweights}, which shows how for a given edge $\{u,v\}\in G$, the $k$-dimensional up-down walk on a graph-product-based construction such as $Z$ proceeds within the faces $\bigcup_{j=1}^kZ((j,k+1-j)_{(u,v)})$ analagously to a random walk on a $k$-vertex path, which has a mixing time of $\Omega(k^2)$.\footnote{For graph-product-based constructions other than $Z$, the analagous path graph may have arbitrary edge weights. We believe that such paths have mixing time $\Omega(k^2)$, but we have not proven such a bound.} These observations suggest that no graph-product-based construction obtains better than a quadratic dependence on $k$ in the spectral gap of the $k$-dimensional up-down walk. It is an interesting open problem to develop alternative combinatorial constructions of high-dimensional expanders that attain the optimal up-down walk spectral gap of $\Omega(\frac{1}{k})$.

\section{Acknowledgements}
The author thanks Salil Vadhan for numerous helpful comments and discussions.

\bibliography{library}

\begin{thebibliography}{RVW02}

\bibitem[AL20]{alev_improved_2020}
Vedat~Levi Alev and Lap~Chi Lau.
\newblock Improved analysis of higher order random walks and applications.
\newblock In {\em Proceedings of the 52nd {Annual} {ACM} {SIGACT} {Symposium}
  on {Theory} of {Computing}}, {STOC} 2020, pages 1198--1211, New York, NY,
  USA, June 2020. Association for Computing Machinery.

\bibitem[CLP20]{chapman_expander_2020}
Michael Chapman, Nati Linial, and Yuval Peled.
\newblock Expander {Graphs} — {Both} {Local} and {Global}.
\newblock {\em Combinatorica}, 40(4):473--509, August 2020.

\bibitem[Con19]{conlon_hypergraph_2019}
David Conlon.
\newblock Hypergraph expanders from {Cayley} graphs.
\newblock {\em Israel Journal of Mathematics}, 233(1):49--65, August 2019.

\bibitem[CTZ20]{conlon_hypergraph_2020}
David Conlon, Jonathan Tidor, and Yufei Zhao.
\newblock Hypergraph expanders of all uniformities from {Cayley} graphs.
\newblock {\em Proceedings of the London Mathematical Society},
  121(5):1311--1336, 2020.

\bibitem[DK17]{dinur_high_2017}
I.~Dinur and T.~Kaufman.
\newblock High {Dimensional} {Expanders} {Imply} {Agreement} {Expanders}.
\newblock In {\em 2017 {IEEE} 58th {Annual} {Symposium} on {Foundations} of
  {Computer} {Science} ({FOCS})}, pages 974--985, October 2017.
\newblock ISSN: 0272-5428.

\bibitem[FI20]{friedgut_hyper-regular_2020}
Ehud Friedgut and Yonatan Iluz.
\newblock Hyper-regular graphs and high dimensional expanders.
\newblock {\em arXiv:2010.03829 [math]}, October 2020.
\newblock arXiv: 2010.03829.

\bibitem[Gro10]{gromov_singularities_2010}
Mikhail Gromov.
\newblock Singularities, {Expanders} and {Topology} of {Maps}. {Part} 2: from
  {Combinatorics} to {Topology} {Via} {Algebraic} {Isoperimetry}.
\newblock {\em Geometric and Functional Analysis}, 20(2):416--526, August 2010.

\bibitem[HLW06]{hoory_expander_2006}
Shlomo Hoory, Nathan Linial, and Avi Wigderson.
\newblock Expander graphs and their applications.
\newblock {\em Bulletin of the American Mathematical Society}, 43(4):439--561,
  2006.

\bibitem[KM17]{kaufman_high_2017}
Tali Kaufman and David Mass.
\newblock High {Dimensional} {Random} {Walks} and {Colorful} {Expansion}.
\newblock In {\em 8th {Innovations} in {Theoretical} {Computer} {Science}
  {Conference} ({ITCS} 2017)}, volume~67 of {\em Leibniz {International}
  {Proceedings} in {Informatics} ({LIPIcs})}, pages 4:1--4:27, Dagstuhl,
  Germany, 2017. Schloss Dagstuhl–Leibniz-Zentrum fuer Informatik.

\bibitem[KO18]{kaufman_high_2018}
Tali Kaufman and Izhar Oppenheim.
\newblock High {Order} {Random} {Walks}: {Beyond} {Spectral} {Gap}.
\newblock In {\em Approximation, {Randomization}, and {Combinatorial}
  {Optimization}. {Algorithms} and {Techniques} ({APPROX}/{RANDOM} 2018)},
  volume 116 of {\em Leibniz {International} {Proceedings} in {Informatics}
  ({LIPIcs})}, pages 47:1--47:17, Dagstuhl, Germany, 2018. Schloss
  Dagstuhl–Leibniz-Zentrum fuer Informatik.

\bibitem[KO20]{kaufman_high_2020}
Tali Kaufman and Izhar Oppenheim.
\newblock High dimensional expanders and coset geometries.
\newblock {\em arXiv:1710.05304 [math]}, May 2020.
\newblock arXiv: 1710.05304 version: 3.

\bibitem[LM06]{linial_homological_2006}
Nathan Linial and Roy Meshulam.
\newblock Homological {Connectivity} {Of} {Random} 2-{Complexes}.
\newblock {\em Combinatorica}, 26(4):475--487, August 2006.

\bibitem[LMY20]{liu_high-dimensional_2020}
Siqi Liu, Sidhanth Mohanty, and Elizabeth Yang.
\newblock High-{Dimensional} {Expanders} from {Expanders}.
\newblock In {\em 11th {Innovations} in {Theoretical} {Computer} {Science}
  {Conference} ({ITCS} 2020)}, volume 151 of {\em Leibniz {International}
  {Proceedings} in {Informatics} ({LIPIcs})}, pages 12:1--12:32, Dagstuhl,
  Germany, 2020. Schloss Dagstuhl–Leibniz-Zentrum fuer Informatik.

\bibitem[LSV05a]{lubotzky_explicit_2005}
Alexander Lubotzky, Beth Samuels, and Uzi Vishne.
\newblock Explicit constructions of {Ramanujan} complexes of type
  $\tilde{A}_d$.
\newblock {\em European Journal of Combinatorics}, 26(6):965--993, August 2005.

\bibitem[LSV05b]{lubotzky_ramanujan_2005}
Alexander Lubotzky, Beth Samuels, and Uzi Vishne.
\newblock Ramanujan complexes of type $\tilde{A}_d$.
\newblock {\em Israel Journal of Mathematics}, 149(1):267--299, December 2005.

\bibitem[Lub18]{lubotzky_high_2018}
Alexander Lubotzky.
\newblock High {Dimensional} {Expanders}.
\newblock In {\em Proceedings of the {International} {Congress} of
  {Mathematicians} ({ICM} 2018)}, pages 705--730. World Scientific, June 2018.

\bibitem[MW09]{meshulam_homological_2009}
R.~Meshulam and N.~Wallach.
\newblock Homological connectivity of random \textit{k} -dimensional complexes:
  {Homological} {Connectivity} of {Random} {Complexes}.
\newblock {\em Random Structures \& Algorithms}, 34(3):408--417, May 2009.

\bibitem[Opp18]{oppenheim_local_2018}
Izhar Oppenheim.
\newblock Local {Spectral} {Expansion} {Approach} to {High} {Dimensional}
  {Expanders} {Part} {I}: {Descent} of {Spectral} {Gaps}.
\newblock {\em Discrete \& Computational Geometry}, 59(2):293--330, March 2018.

\bibitem[RVW02]{reingold_entropy_2002}
Omer Reingold, Salil Vadhan, and Avi Wigderson.
\newblock Entropy {Waves}, the {Zig}-{Zag} {Graph} {Product}, and {New}
  {Constant}-{Degree} {Expanders}.
\newblock {\em The Annals of Mathematics}, 155(1):157, January 2002.

\end{thebibliography}
\bibliographystyle{alpha}

\end{document}